\documentclass[12pt, reqno]{amsart}
\usepackage{amsmath, amsthm, amscd, amsfonts, amssymb, graphicx, color, mathrsfs}
\usepackage[bookmarksnumbered, colorlinks, plainpages]{hyperref}
\usepackage[all]{xy} 
\usepackage{slashed}
\usepackage{cite}
\usepackage{xcolor}
\usepackage{listings}
\usepackage{xcolor}

\newcommand{\set}[2]{\left\{#1: \, #2\right\}}
   
\newcommand{\sclp}[2][]{\langle#2\rangle} 

\newcommand{\N}{\mathbb{N}}

\newcommand{\R}{\mathbb{R}}
\newcommand{\C}{\mathbb{C}}

\newcommand{\eps}{\varepsilon}
\newcommand{\ii}{\mathrm{i}}
\newcommand{\e}{\mathrm{e}}

\newcommand{\dom}{D} 
 
\newcommand{\dd}{\, \mathrm{d}}

\newtheorem{theorem}{Theorem}[section]

\theoremstyle{definition}

\theoremstyle{remark}
\newtheorem{remark}[theorem]{Remark}
\numberwithin{equation}{section}

\begin{document}
\setcounter{page}{1}

\color{darkgray}{
\noindent \centering
{\small   }\hfill    {\small }\\
{\small }\hfill  {\small }}

\centerline{}

\centerline{}


\title[Functional-difference operators for  mirror curves]{Riesz-means  bounds for functional-difference operators for  mirror curves}

\author[D. Cardona]{Duv\'an Cardona$^1$}
\address{
  Duv\'an Cardona:
  \endgraf
  Department of Mathematics: Analysis, Logic and Discrete Mathematics
  \endgraf
  Ghent University, Belgium
  \endgraf
  {\it E-mail address} {\rm duvanc306@gmail.com, duvan.cardonasanchez@ugent.be}
  }

\thanks{{$^{1}$ {Website:}
\url{https://sites.google.com/site/duvancardonas/home}}
\newline   Duv\'an Cardona is supported  by the FWO  Odysseus  1  grant  G.0H94.18N:  Analysis  and  Partial Differential Equations and by the Methusalem programme of the Ghent University Special Research Fund (BOF)
(Grant number 01M01021). Duván Cardona has been supported by the FWO Fellowship
Grant No 1204824N and by the FWO Grant K183725N of the Belgian Research Foundation FWO}

\begin{abstract} Let \( P \) and \( Q \) be the quantum-mechanical momentum and position operators on \( L^2(\R) \).
Let $\zeta>0.$ We provide estimates  for the {\it Riesz means} $\varkappa(\lambda)$ associated with the system of eigenvalues of the operator
 \begin{align}
H(\zeta) = \e^{-bP} + \e^{bP} + \e^{2\pi b Q} + \zeta \e^{-2\pi b Q} = U + U^{-1} + V + \zeta V^{-1},
\end{align} when $\lambda\rightarrow\infty.$  This operator arises in the quantisation of the local {\it del Pezzo Calabi-Yau threefold}, defined
as the total space of the anti-canonical bundle over the {\it Hirzebruch surface}  \( S = \mathbb{P}^{1} \times \mathbb{P}^{1} \).   Our approach is motivated by the spectral analysis of $\varkappa(\lambda)$ in the framework developed by Laptev, Schimmer and Takhtajan in \cite{Laptev2016}.
\end{abstract} \maketitle
\allowdisplaybreaks
\tableofcontents


\section{Introduction}
\subsection{Overview} Let \( P \) and \( Q \) be the quantum-mechanical {\it momentum} and {\it position} operators on \( L^2(\R) \), defined on a common domain where they satisfy the Heisenberg commutation relation \( [P, Q]:=PQ-QP = \ii I \). Consider the associated {\it Weyl operators} \( U = \e^{-bP} \) and \( V = \e^{2\pi b Q} \), with \( b > 0 \). These are unbounded, self-adjoint operators on \( L^2(\R) \), and on their joint domain they satisfy the Weyl relation:
\begin{equation}\label{U:V}
    UV = q^2 VU,
\end{equation}
where \( q = \e^{\ii\pi b^2} \). From the perspective of the mathematical physics, an interesting fact was discovered in \cite{ADKMV}, where some operators constructed from the Weyl operators \( U \) and \( V \) were shown to arise in the context of {\it local mirror symmetry}, as quantizations of {\it algebraic curves mirror} to {\it toric Calabi-Yau threefolds}. Their spectral properties have been of great interest after the development built in \cite{GHM}. A typical example is provided by a local {\it del Pezzo Calabi-Yau threefold}, defined as the total space of the anti-canonical bundle over a toric {\it del Pezzo surface} \( S \). 

Moreover, when \( S = \mathbb{P}^{1} \times \mathbb{P}^{1} \)  is the Hirzebruch surface, one obtains the operator
\begin{align} \label{eq:type1}
H(\zeta) = \e^{-bP} + \e^{bP} + \e^{2\pi b Q} + \zeta \e^{-2\pi b Q} = U + U^{-1} + V + \zeta V^{-1},
\end{align}
where \( \zeta > 0 \) is a {\it mass parameter}, so that \( H = H(0) \), is the
functional-difference operator $H,$ defined on $L^2(\mathbb{R}),$  by
 $$ 
(H\psi)(x) = \psi(x + \ii b) + \psi(x - \ii b) + \e^{2\pi b x} \psi(x).
 $$ We observe that if instead of \( S \) one considers the weighted projective space \( \mathbb{P}(1,m,n) \) with \( m, n \in \N \), the corresponding operator is
\begin{align} \label{eq:type2}
H_{m,n} = \e^{-bP} + \e^{2\pi b Q} + \e^{bmP - 2\pi bnQ} = U + V + q^{-mn} U^{-m} V^{-n},
\end{align}
and again \( H = H_{1,0} \) (see \cite{GHM} for further details).

As noted in \cite{Laptev2016}, it was conjectured in \cite{GHM} that for \( \zeta > 0 \) and \( m, n \in \N \), the operators \( H(\zeta) \) and \( H_{m,n} \) have discrete spectra, their inverses are of trace class, and their Fredholm determinants admit explicit expressions in terms of enumerative invariants of the corresponding Calabi–Yau threefolds. Some of these conjectures were proven in \cite{Kashaev2015}, where explicit formulas for \( H(\zeta)^{-1} \) and \( H^{-1}_{m,n} \) were obtained in terms of the modular quantum dilogarithm. Our paper is motivated by the analysis in Laptev, Schimmer and Takhtajan in \cite{Laptev2016} where the authors have shown that these operators have  purely discrete spectrum and have proved that they have many other interesting properties. In particular that their Weyl fuctions $N(\lambda)$ grow as $\log^2(\lambda),$ when $\lambda\rightarrow\infty.$ We are interested in the distribution of the eigenvalues of $H(\zeta).$ In the next subsection we present our main result, see Theorem \ref{theo1}. In what follows for any $t\in \mathbb{R,}$ we use the notation $(t)_+:=\max\{t,0\}=(t+\vert t\vert)/2.$

\subsection{Notation} We denote by $A\asymp B,$ if we have inequalities $A\leq C B$ and $cA\geq B,$ with $C,c>0,$ independent of fundamental parameters. In particular, we use $A\lesssim B,$ if $A\leq CB,$ and $A\gg B$ if $A\geq cB,$ again with $C,c>0,$ independent of fundamental parameters.

\subsection{Main result}
In this paper we are interested in investigating the distribution of the eigenvalues of the operator $H(\zeta),$ with $\zeta>0.$ Indeed,  Laptev, Schimmer and Takhtajan in \cite{Laptev2016} proved that  the Riesz-means $$\varkappa(\lambda):=\sum_{j\geq 1}(\lambda-\lambda_j)_{+}$$  associated to the eigenvalues $\lambda_j,$ of $H(\zeta),$
obeys the asymptotic behaviour 
\begin{equation}\label{Laptev:2016:Weyl}
    \varkappa(\lambda)\asymp\frac{1}{\pi^2b^2}\lambda\log^2(\lambda)+O(\lambda\log(\lambda)),\,\quad \lambda\rightarrow \infty,
\end{equation} in the sense that for some constants $C_1,C_2\in \mathbb{R},$
\begin{equation}\label{aux} \frac{1}{\pi^2b^2}\lambda\log^2(\lambda)+C_1(\lambda\log(\lambda))    \leq \varkappa(\lambda)\leq \frac{1}{\pi^2b^2}\lambda\log^2(\lambda)+C_2(\lambda\log(\lambda)),\,\quad \lambda\rightarrow \infty.
\end{equation}
 Given continuity to the techniques from \cite{Laptev2016}, our approach reveals that a suitable {\it dyadic decomposition} of the phase space, see  \eqref{Regions} in Remark \ref{remark:intro}, shows evidence of the existence of other 5 lower orders terms in the behaviour of $\varkappa(\lambda)$. These lower order terms could involve,  when $\lambda\rightarrow\infty,$  terms with sizes, $\log (\lambda),$  $\log^2 (\lambda),$ $-\lambda\log(\lambda),$ $(\lambda-\log_2(\lambda))\log (\lambda-\log_2(\lambda)),$ and also of size $-(\lambda-\log_2(\lambda))\log^2 (\lambda-\log_2(\lambda)).$
The following is our main result.
\begin{theorem}\label{theo1}Let $\zeta>0.$ We have the following inequality for the Riesz-means of the eigenvalues $\lambda_j$ of $H(\zeta),$
\begin{align*}
  \varkappa(\lambda) =\sum_{j\geq 1}(\lambda-\lambda_j)_{+} \leq \frac{1}{\pi^2b^2}(\lambda \log^2(\lambda)),
\end{align*}    when $\lambda$ is large enough. Moreover,
\begin{align*}
    \varkappa(\lambda) &\asymp  \frac{1}{\pi^2b^2}(\lambda \log^2(\lambda)-(\lambda- \log_2(\lambda))\log^2(\lambda-\log_2(\lambda))\\
 &-2\lambda\log(\lambda)+2(\lambda-\log_2(\lambda))\log(\lambda-\log_2(\lambda))+\log^2(\lambda)+\log_2(\lambda)),
\end{align*}  when $\lambda$ is large enough.
\end{theorem}
\begin{remark}
   In what follows, we illustrate the size of $\varkappa.$
\begin{figure}[ht]
    \centering
    \includegraphics[width=0.45\textwidth]{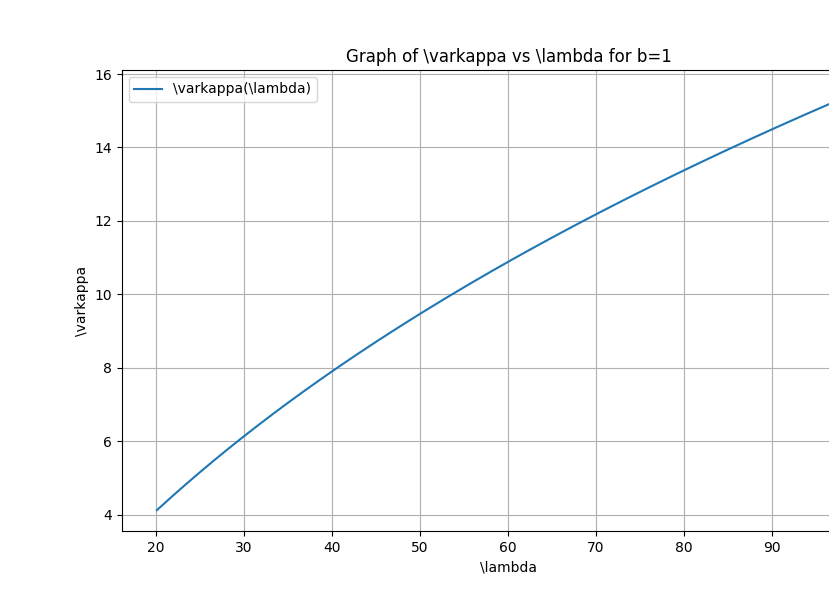}\includegraphics[width=0.45\textwidth]{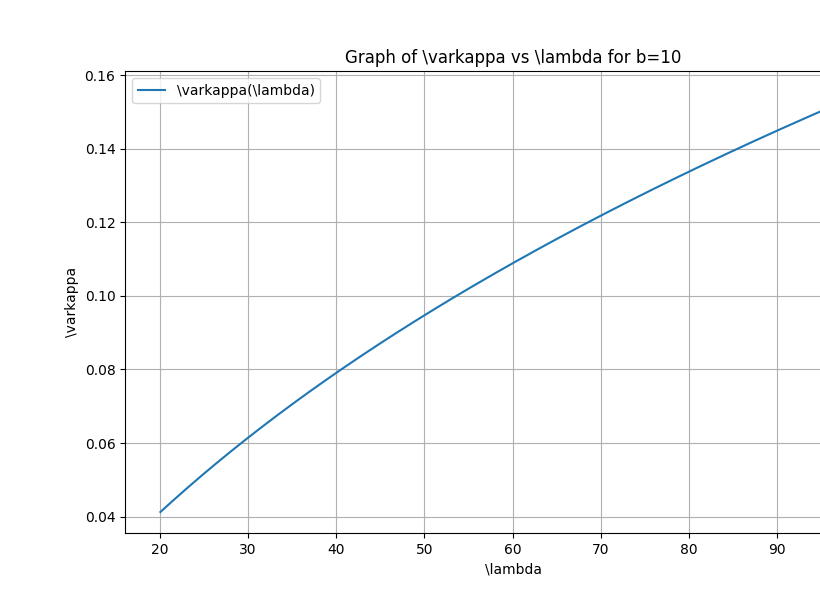}
    \caption{Graphs of the Riesz means $\varkappa$ for $b=1,$ and for $b=10,$ respectively, in the interval $[20.1,97]$.}
    \label{fig:region-cosh}
\end{figure}
\end{remark}
\subsection{Methodology} Let us briefly discuss our methodology for the proof of Theorem \ref{theo1} in Remark \ref{remark:intro} below. 
\begin{figure}[ht]
    \centering
    \includegraphics[width=0.55\textwidth]{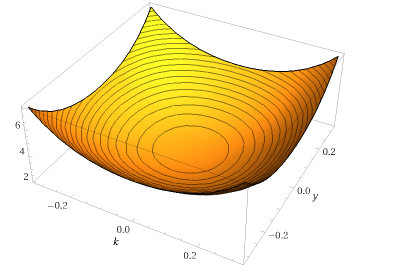}\includegraphics[width=0.42\textwidth]{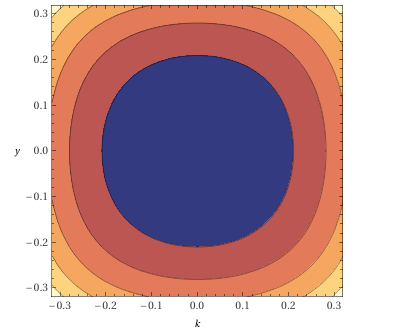}
    \caption{The surface \( z = \cosh(2\pi k) + \cosh(2\pi y) \) in \( \mathbb{R}^3 \) and its level curves in the \( (k, y) \) plane.}
    \label{fig:region-cosh}
\end{figure}
\begin{remark}\label{remark:intro} Our approach is motivated by the techniques in \cite[Chapter 7]{CDR:AA}, see also \cite{CDR:JMP}, which can be traced back to Delgado \cite{Del2006}.
    As it was observed in Remark 2.4 in \cite{Laptev2016},  the leading term $$\lambda\log^2(\lambda)/(\pi b)^2,$$ coincides with the leading term in the phase
space integral
$$\iint_{\mathbb{R}^2}(\lambda-\sigma(y,k))_{+}dk dy,\quad \lambda\rightarrow\infty,$$
 where $\sigma(y,k)$ is the pseudo-differential symbol of the operator $H(\zeta).$  By taking as a model the dyadic decompositions of symbols in harmonic analysis, we will apply in our further analysis a dyadic decomposition of the phase space for the corresponding integrals estimating the  Riesz-means $\varkappa(\lambda)$ of $H(1).$ Indeed, in view of Remark \ref{main:remark}, our approach forces us to understand the geometry of the regions:
 \begin{equation}\label{Regions}
    \Omega_{j}=\{(y,k):2^{j-1}\leq \lambda-2d\cosh(2\pi bk)- 2d\cosh(2\pi by)\leq 2^j \},\quad j\in \mathbb{Z}.
\end{equation} There are exactly two classes of sets where the contributions of the volume of the sets $\Omega_j$ produces differents lower order terms. The information of all those $\Omega_j$ with $j<0$  provides a $\log(\lambda)$ error term. The geometry of the sets $\Omega_j$ with $1\leq j\leq\frac{\log(\lambda-4d)}{\log 2}+1,$ encodes the information of the leading term $\lambda\log^2(\lambda)$ and also of the other lower order terms. An interesting fact of the sets $\Omega_j$ is that they look like {\it twisted annuli}, see Figure \ref{fig:region-cosh}.
\end{remark}
\begin{remark} Let us consider the Weyl eigenvalue counting formula $N:=N(\lambda)$ of the operator $H(\zeta).$
    One can use the analysis in \cite[Page 297]{Laptev2016} in order to deduce that
    \begin{equation}
        \frac{\varkappa(\lambda)}{\lambda}\leq  N(\lambda)\leq \frac{\varkappa((\tau+1)\lambda)}{\tau\lambda}\leq \frac{1}{\pi^2b^2}(1+\tau^{-1}) \log^2((1+\tau)\lambda),\,\quad \forall\tau>0.
    \end{equation} 
\end{remark}
\begin{remark}
    Our approach can be also adapted for the family of operators $H_{m,n}$ in \eqref{eq:type2}. However, we omit such an analysis here. 
\end{remark}
\subsection{Organisation of the paper}
This paper is organised as follows. In Section \ref{Preliminaries}, we provide the preliminaries related to the spectral theory of the operator $H(\zeta),$ according to the analysis in Laptev, Schimmer and Takhtajan in \cite{Laptev2016}. The proof of Theorem \ref{theo1} is presented in Section \ref{Proof:theo}.

\section{Preliminaries}\label{Preliminaries}

In this section we present the preliminaries about the construction of the operators $H(\zeta).$ For the general aspects of the spectral theory of self-adjoint operators on Hilbert spaces, we refer the reader, e.g., to 
 Birman and  Solomjak \cite{Birman1987}. For the spectral theory of $H(\zeta)$ we follow \cite{Laptev2016}.
\subsection{Weyl operators}
Let us consider the standard Lebesgue spaces $ L^1(\R)$ and $ L^2(\R),$ defined by
\begin{equation*}
    f\in L^1(\mathbb{R})\Longleftrightarrow  f:\mathbb{R}\rightarrow \mathbb{C}\textnormal{ is Lebesgue measurable, and }\smallint_{\mathbb{R}}|f(x)|dx<\infty,
\end{equation*}
\begin{equation*}
    \psi\in L^2(\mathbb{R})\Longleftrightarrow  \psi:\mathbb{R}\rightarrow \mathbb{C}\textnormal{ is Lebesgue measurable, and }\smallint_{\mathbb{R}}|\psi(x)|^2dx<\infty,
\end{equation*} where $dx$ denotes the Lebesgue measure on the $\sigma$-algebra of Borel sets in $\mathbb{R}.$

Let us
 consider the associated Weyl operators \( U = \e^{-bP} \) and \( V = \e^{2\pi b Q} \), with \( b > 0 \). These are unbounded, self-adjoint operators on \( L^2(\R) \). In the coordinate representation, the action of the operators \( P \) and \( Q \) is given by \( (P\psi)(x) = \ii \psi'(x) \) and \( (Q\psi)(x) = x\psi(x) \). The Weyl operators act as \( (U\psi)(x) = \psi(x + \ii b) \) and \( (V\psi)(x) = \e^{2\pi b x} \psi(x) \). Their respective domains are characterized by:
\begin{align*}
\dom(U) &= \set{\psi \in L^2(\R)}{\e^{-2\pi b k} (\mathscr{F}{\psi})(k) \in L^2(\R)}, \\
\dom(V) &= \set{\psi \in L^2(\R)}{\e^{2\pi b x} \psi(x) \in L^2(\R)},
\end{align*}
where \( \mathscr{F} \) denotes the Fourier transform, densely defined for $f\in L^1(\mathbb{R})\cap L^2(\mathbb{R})$ by:
\begin{align*}
\widehat{f}(k) = (\mathscr{F} f)(k) = \smallint_\R \e^{-2\pi \ii k x} f(x) \dd x.
\end{align*}
We also remark that alternatively, the domain \( \dom(U) \) can be described as the set of functions \( \psi(x) \) that extend analytically to the strip $$  \set{z = x + \ii y \in \C}{0 < y < b}, $$ with \( \psi(x + \ii y) \in L^2(\R) \) for all \( 0 \leq y < b \), and such that the limit
\[
\psi(x + \ii b - \ii 0) = \lim_{\eps \to 0^+} \psi(x + \ii b - \ii \eps),
\]
exists with respect to the topology induced by the norm in \( L^2(\R) \). We denote this limit simply by \( \psi(x + \ii b) \). The domain of the inverse operator \( U^{-1} \) admits a similar characterisation.
\subsection{The spectral analysis of Laptev, Schimmer and Takhtajan for $H(\zeta)$}
Let us use the Weyl operators \( U \) and \( V \) in \eqref{U:V} in order to define
\[
H = U + U^{-1} + V,
\]
which, in the coordinate representation, becomes a functional-difference operator with exponential potential given by
\[
(H\psi)(x) = \psi(x + \ii b) + \psi(x - \ii b) + \e^{2\pi b x} \psi(x).
\] 
As a twisted verion of $H,$ for any $\zeta>0,$ we consider the operators \begin{align} \label{eq:type1:8}
H(\zeta) = \e^{-bP} + \e^{bP} + \e^{2\pi b Q} + \zeta \e^{-2\pi b Q} = U + U^{-1} + V + \zeta V^{-1}.
\end{align}
\begin{remark}As it was mentioned in the introduction, it was conjectured in \cite{GHM} that for \( \zeta > 0 \) and \( m, n \in \N \), the operators \( H(\zeta) \) and \( H_{m,n} \) have discrete spectra, their inverses are of trace class, and their Fredholm determinants admit explicit expressions in terms of enumerative invariants of the corresponding Calabi–Yau threefolds. Some of these conjectures were proven in the recent work \cite{Kashaev2015}, where explicit formulas for \( H(\zeta)^{-1} \) and \( H^{-1}_{m,n} \) were obtained in terms of the modular quantum dilogarithm.    
\end{remark}

\begin{remark}
    The analysis in   \cite{Laptev2016} has obtained  Weyl-type asymptotics for the operators \( H(\zeta) \) and \( H_{m,n} \), considered as self-adjoint operators on \( L^2(\R) \). Specifically, the authors in  \cite{Laptev2016} have established that both operators have purely discrete spectra and have analysed the asymptotic behavior of their eigenvalues. As a consequence of their analysis, it follows that \( H(\zeta)^{-1} \) and \( H_{m,n}^{-1} \) are of trace class.
\end{remark}

\begin{remark} The  main results in \cite{Laptev2016}, which describe the asymptotic behavior of the Riesz mean \( \sum_{j \ge 1} (\lambda - \lambda_j)_+ \), and also state the corresponding Weyl laws for the eigenvalue counting function \( N(\lambda) \) associated with these operators, prove that
\begin{equation} \label{W-zeta}
\lim_{\lambda \rightarrow \infty} \frac{N(\lambda)}{\log^2 \lambda} = \frac{1}{(\pi b)^2},
\end{equation}
for the operator \( H(\zeta) \), and
\begin{equation} \label{W-mn}
\lim_{\lambda \rightarrow \infty} \frac{N(\lambda)}{\log^2 \lambda} = \frac{c_{m,n}}{(2\pi b)^2}, \quad c_{m,n} = \frac{(m+n+1)^2}{2mn},
\end{equation}
for the operator \( H_{m,n} \). We remark that the proof in \cite{Laptev2016} follows ideas originally developed in \cite{Laptev1997}, where the Fourier transform is replaced by the coherent state transform. The techniques employed in \cite{Laptev2016} also draw inspiration from the derivation of the Berezin–Lieb inequality \cite{Berezin1972,Berezin1972b,Lieb1973}.    
\end{remark}

\begin{remark} The operator $H$ first appeared in the context of the quantum Liouville model on the lattice \cite{Faddeev1986}, and it plays a significant role in the representation theory of the non-compact quantum group \( \mathrm{SL}_{q}(2, \R) \). In the momentum representation, the operator \( H \) corresponds to the Dehn twist operator in quantum Teichm\"uller theory \cite{Kashaev2001}. In particular, \cite{Kashaev2001} formulated the eigenfunction expansion theorem for \( H \) in the momentum representation as formal completeness and orthogonality relations, interpreted in the distributional sense.    
\end{remark}

\begin{remark}
    The spectral analysis of $H$ was carried out in \cite{Faddeev2014}, where it was shown that \( H \) is self-adjoint with purely absolutely continuous spectrum equal to the interval \( [2, \infty) \). 
\end{remark}
\subsubsection{The Coherent State Representation}\label{subsec:coherent}
Let $g$ be the Gaussian function $$g(x)=(a/\pi)^{1/4}\e^{-\frac a2 x^2},$$ where  $a>0$. Observe that $g$ belongs to the unit ball in $L^2(\R)$. 
For $\psi\in L^2(\R)$ the classical coherent state transform (see e.g. \cite[Chapter 12]{Lieb2001}) is given by
\begin{align*}
\widetilde{\psi}(k,y)=\int_\R\e^{-2\pi\ii k x}g(x-y)\psi(x)\dd x\,.
\end{align*} For $\psi\in\dom(H_0)$ we have  the representation 
\begin{align}
\sclp[2]{H_0\psi,\psi}=d_1\iint_{\R^{2}} 2\cosh(2\pi b k)|\widetilde{\psi}(k,y)|^2\dd k\dd y,
\label{eq:rep1} 
\end{align}
where
$$d_1=\e^{-ab^2/4}<1.$$ Similarly, for $\psi\in\dom(W),$  where $W=2\cosh(2\pi b x),$ we have the representation 
\begin{align}
\sclp[2]{W\psi,\psi}=d_2\iint_{\R^{2}} 2\cosh(2\pi b y)|\widetilde{\psi}(k,y)|^2\dd k\dd y,
\label{eq:rep2} 
\end{align} where $d_2=e^{-(\pi b)^2/a}.$
 Then, one has the identity, see \cite{Laptev2016},
\begin{equation} \label{H-coherent}
\sclp[2]{H(1)\psi,\psi}=\iint_{\R^{2}}2(d_1\cosh(2\pi b k)+d_2\cosh(2\pi by))|\widetilde{\psi}(k,y)|^2\dd k\dd y.
\end{equation}
In the next section we investigate the asymptotics for the Riesz-means of the oeprator $H(\zeta).$

\section{Proof of the main theorem}\label{Proof:theo}

Let $\zeta>0.$ Let us consider the Weyl operators $U$ and $V,$ and then the operator \begin{align} \label{eq:type1:2}
H(\zeta) = \e^{-bP} + \e^{bP} + \e^{2\pi b Q} + \zeta \e^{-2\pi b Q} = U + U^{-1} + V + \zeta V^{-1}.
\end{align} Following the strategy in  \cite{Laptev2016}, in order to prove Theorem \ref{theo1}, we restrict the analysis of the Riesz means of $H(\zeta),$  just to the case $\zeta=1.$ From now we fix 
\begin{equation}
   \varkappa(\lambda)= \sum_{j\geq 1}(\lambda-\lambda_j)_+,
\end{equation} where $\{\lambda_j\}_{\,j\geq 1},$ denotes the system of eigenvalues of the operator $H(1).$ We recall that for any $t\in \mathbb{R,}$ $(t)_+:=\max\{t,0\}=(t+\vert t\vert)/2.$ 
\begin{remark} As it was observed in Remark 2.4 in \cite{Laptev2016}, and mentioned in the introduction, the leading term $\lambda\log^2(\lambda)/(\pi b)^2$ coincides with the leading term in the phase
space integral
$$\iint_{\mathbb{R}^2}(\lambda-\sigma(y,k))_{+}dk dy,\quad \lambda\rightarrow\infty,$$
 where $\sigma(y,k)$ is the pseudo-differential symbol of the operator $H(\zeta).$  By taking as a model the dyadic decompositions of symbols in harmonic analysis, we will apply in our further analysis a Littlewood-Paley partition of the phase space for the corresponding integrals estimating the  Riesz-means $\varkappa(\lambda)$ of $H(1).$
 \end{remark}
 \begin{remark}\label{main:remark}
      We start by considering the following upper and lower bounds for $\varkappa(\lambda)$, see \cite[Pages 294 and 296]{Laptev2016}, which are a consequence of the identity in \eqref{H-coherent},
\begin{align*}
\sum_{j\geq 1}(\lambda-\lambda_j)_+\leq I_1= 4{\iint}_{\mathbb{R}^+_0\times \mathbb{R}^+_0} \big(\lambda-2d_1\cosh(2\pi bk)- 2d_2\cosh(2\pi by)\big)_+d kd y, 
\end{align*} and 
\begin{align*}
\sum_{j\geq 1}(\lambda-\lambda_j)_+ 
&\geq I_2= 4 \iint_{\mathbb{R}^+_0\times \mathbb{R}^+_0}\left(\lambda-\frac{2}{d_1}\cosh(2\pi bk)-\frac{2}{d_2}\cosh(2\pi by)\right)_+d k d y. 
\end{align*} Observe the presence of the factor $4$ in the right hand side of these inequalities in view of the symmetry of the surface $z = \cosh(2\pi k) + \cosh(2\pi y).$
\begin{figure}[ht]
    \centering
    \includegraphics[width=0.45\textwidth]{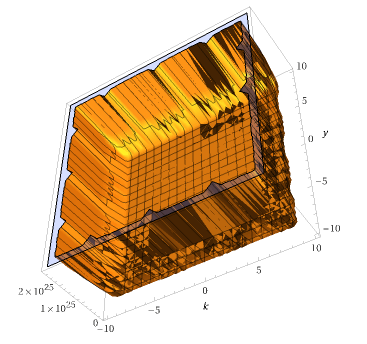}\includegraphics[width=0.40\textwidth]{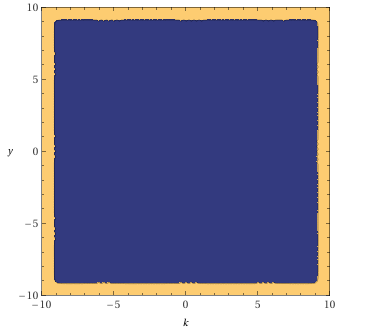}
    \caption{Symmetry of the surface \( z = \cosh(2\pi k) + \cosh(2\pi y) \) in \( \mathbb{R}^3 \) on the domain \( \{(k, y):-10\leq k,y\leq 10\} \).}
    \label{fig:region-cosh:2}
\end{figure} 
 \end{remark}
 
\begin{proof}[Proof of Theorem \ref{theo1}]
Consider the integral
\[
I := \iint_{\mathbb{R}^+_0 \times \mathbb{R}^+_0} \big(\lambda - 2d_1 \cosh(2\pi b k) - 2d_2 \cosh(2\pi b y)\big)_+ \, dk \, dy,
\]
where \(d_1, d_2, b > 0\) are as above and \(\lambda\) is such that the integrand is not identically zero, and in particular such that $\lambda>2d_1+2d_2$.
\subsection{The upper bound $(\pi b)^{-2}\lambda\log^2\lambda$}
Observe that for $k,y\geq 0,$ we have the following analysis. 
The integrand is positive in the region of the $(k,y)$-plane where
\[
\lambda - 2d_1 \cosh(2\pi b k) - 2d_2 \cosh(2\pi b y) > 0.
\]
Fix \(k \ge 0\) and  $\lambda$ large enough. Then the condition in \(y\) becomes
\[
\cosh(2\pi b y) < \frac{\lambda - 2d_1 \cosh(2\pi b k)}{2d_2}.
\]
Since \(\cosh(2\pi b y) \geq 1\), it is necessary that
\[
\lambda - 2d_1 \cosh(2\pi b k) > 2d_2,
\]
which defines the domain in \(k\):
\[
0 \leq k < k_0 := \frac{1}{2\pi b} \operatorname{arccosh}\left( \frac{\lambda - 2d_2}{2d_1} \right).
\]
For fixed \(k\), define
\[
y_{\max}(k) := \frac{1}{2\pi b} \operatorname{arccosh}\left( \frac{\lambda - 2d_1 \cosh(2\pi b k)}{2d_2} \right).
\]
Then,
\[
\begin{aligned}
&\int_0^{y_{\max}(k)} \big(\lambda - 2d_1 \cosh(2\pi b k) - 2d_2 \cosh(2\pi b y)\big) \, dy\\
&= \left(\lambda - 2d_1 \cosh(2\pi b k)\right) y_{\max}(k)  - \frac{2 d_2}{2\pi b} \sinh\big(2\pi b y_{\max}(k)\big).
\end{aligned}
\]
Using the identity \(\sinh(\operatorname{arccosh}(u)) = \sqrt{u^2 - 1}\) and setting
\[
u(k) := \frac{\lambda - 2d_1 \cosh(2\pi b k)}{2 d_2},
\]
we have
\[
y_{\max}(k) = \frac{1}{2\pi b} \operatorname{arccosh} u(k), \quad \sinh(2\pi b y_{\max}(k)) = \sqrt{u(k)^2 - 1}.
\]
Hence,
\[
I \leq  \int_0^{k_0} \frac{1}{2\pi b} \left[ \big(\lambda - 2d_1 \cosh(2\pi b k)\big) \operatorname{arccosh} u(k) - 2 d_2 \sqrt{u(k)^2 - 1} \right] dk,
\]
where
\[
k_0 = \frac{1}{2\pi b} \operatorname{arccosh} \left( \frac{\lambda - 2 d_2}{2 d_1} \right).
\]
Using the change  of variables \(u = \cosh(2\pi b k)\), with
\[
dk = \frac{1}{2\pi b} \frac{du}{\sqrt{u^2 - 1}},
\]
and the integration limits
\[
u \in \left[1, \frac{\lambda - 2 d_2}{2 d_1}\right],
\]
we get
\[
I \leq  \frac{1}{(2\pi b)^2} \int_1^{\frac{\lambda - 2 d_2}{2 d_1}} \frac{
(\lambda - 2 d_1 u) \operatorname{arccosh}\left(\frac{\lambda - 2 d_1 u}{2 d_2}\right)
- 2 d_2 \sqrt{\left(\frac{\lambda - 2 d_1 u}{2 d_2}\right)^2 - 1}
}{\sqrt{u^2 - 1}} \, du
\]
\[
\leq  \frac{1}{(2\pi b)^2} \int_1^{\frac{\lambda - 2 d_2}{2 d_1}} \frac{
1}{\sqrt{u^2 - 1}}  \, du   \times \lambda \operatorname{arccosh}\left(\frac{\lambda }{2 d_2}\right) 
\]
\[
\leq  \frac{1}{(2\pi b)^2} \operatorname{arccosh}\left(\frac{\lambda-2d_2 }{2 d_1}\right)   \times \lambda \operatorname{arccosh}\left(\frac{\lambda }{2 d_2}\right)
\]
\[
\leq  \frac{1}{(2\pi b)^2} \operatorname{log}\left(\frac{\lambda-2d_2 }{ d_1}\right)   \times \lambda \operatorname{log}\left(\frac{\lambda }{ d_2}\right).
\]
In consequence we have proved that
\begin{align*}
\sum_{j\geq 1}(\lambda-\lambda_j)_+\leq I_1= &4{\iint}_{\mathbb{R}^+_0\times \mathbb{R}^+_0} \big(\lambda-2d_1\cosh(2\pi bk)- 2d_2\cosh(2\pi by)\big)_+d kd y\\
&\leq \frac{1}{(\pi b)^2} \operatorname{log}\left(\frac{\lambda-2d_2 }{ d_1}\right)   \times \lambda \operatorname{log}\left(\frac{\lambda }{ d_2}\right).
\end{align*} In consequence, when $\lambda $ is large enough, we get 
\begin{equation}
    \sum_{j\geq 1}(\lambda-\lambda_j)_+\leq \frac{1}{\pi^2b^2}\lambda\log^2\lambda.
\end{equation}
\subsection{The dyadic analysis of the phase space: an upper bound}
Now, we are going to prove that
\begin{align*}
    \varkappa(\lambda) &\asymp  \frac{1}{\pi^2b^2}(\lambda \log^2(\lambda)-(\lambda- \log_2(\lambda))\log^2(\lambda-\log_2(\lambda))\\
 &-2\lambda\log(\lambda)+2(\lambda-\log_2(\lambda))\log(\lambda-\log_2(\lambda))\\
     &+\log^2(\lambda)+\log_2(\lambda)).
\end{align*}
Let us analyse $I_1. $ So, let us take $d=\min\{d_1,d_2\},$ let us assume that $\lambda>4d-1,$ and let us use that  
\begin{align*}
I_1 &\leq 4 \iint_{\mathbb{R}^+_0\times \mathbb{R}^+_0}  \big(\lambda-2d_1\cosh(2\pi bk)- 2d_2\cosh(2\pi by)\big)_+d kd y \\
&\leq 4 \iint_{\mathbb{R}^+_0\times \mathbb{R}^+_0} \big(\lambda-2d\cosh(2\pi bk)- 2d\cosh(2\pi by)\big)_+d kd y.
\end{align*} Let us adopt the perspective of the Littlewood-Paley theory for pseudo-differential operators, and let us decompose the set $\mathbb{R}^+_0\times \mathbb{R}^+_0$ in the family of dyadic regions  
\begin{equation}
    \Omega_{j}=\{(y,k):2^{j-1}\leq \lambda-2d\cosh(2\pi bk)- 2d\cosh(2\pi by)\leq 2^j \},\quad j\in \mathbb{Z}.
\end{equation}
Then, we have that
\begin{align*}
    I_1 &\leq 4  \iint_{\mathbb{R}^+_0\times \mathbb{R}^+_0} \big(\lambda-2d\cosh(2\pi bk)- 2d\cosh(2\pi by)\big)_+d kd y\leq  4\sum_{j\in \mathbb{Z}}\iint_{ \Omega_{j}} 2^jdkdy. 
\end{align*}
Note that for any $j\in \mathbb{Z},$ we have that 
\begin{align*}
     \Omega_{j}=\{(y,k):\frac{\lambda-2^j}{2d}\leq \cosh(2\pi bk)+ \cosh(2\pi by)\leq \frac{\lambda-2^{j-1}}{2d} \}.
\end{align*}
Note that
$$\forall (y,k),\quad 2\leq \cosh(2\pi bk)+ \cosh(2\pi by). $$ In consequence, 
$$ \frac{\lambda-2^{j-1}}{2d} <2\Longleftrightarrow j>\frac{\log(\lambda-4d)}{\log 2}+1,\textnormal{ and for these values of } j,\,\quad \Omega_j=\{\emptyset\}.$$
In consequence, the decomposition of $\mathbb{R}^+_0\times \mathbb{R}^+_0$ in the sets $\Omega_j$ runs over $j\in \mathbb{Z}$ such that $j\leq \frac{\log(\lambda-4d)}{\log 2}+1. $

In consequence, the change of variables $(u,v)=(\cosh(2\pi bk),\cosh(2\pi by))$ leads to 

\begin{align*} 
&4\sum_{j\in \mathbb{Z}}\iint_{ \Omega_{j}} 2^jdkdy \\
&= \frac{1}{\pi^2b^2} \sum_{j\leq \frac{\log(\lambda-4d)}{\log 2}+1} 2^j \iint_{  \{(u,v): \frac{\lambda-2^j}{2d}\leq u+v\leq \frac{\lambda-2^{j-1}}{2d},\,u,v\geq 1 \} }\frac{dudv}{\sqrt{u^2-1}\sqrt{v^2-1}}.
\end{align*} Let us analyse the integral
$$ R_j:= \iint_{  \{(u,v): \frac{\lambda-2^j}{2d}\leq u+v\leq \frac{\lambda-2^{j-1}}{2d},\,u,v\geq 1 \} }\frac{dudv}{\sqrt{u^2-1}\sqrt{v^2-1}}.$$
   \begin{item}
   \item Case 1: $\frac{\lambda-2^j}{2d}<1.$ In this case $\frac{\log(\lambda-2d)}{\log 2}<j,$ and with $B=\frac{\lambda-2^{j-1}}{2d},$ we have that
   \begin{align*}
       R_j &= \iint_{  \{(u,v): \frac{\lambda-2^j}{2d}\leq u+v\leq \frac{\lambda-2^{j-1}}{2d},\,u,v\geq 1 \} }\frac{dudv}{\sqrt{u^2-1}\sqrt{v^2-1}}\\
       &\leq \int\limits_{1}^B\frac{1}{\sqrt{u^2-1}}\int\limits_{1}^{B-u}\frac{dvdu}{\sqrt{v^2-1}}\\
       &\leq \int\limits_{1}^B\frac{du}{\sqrt{u^2-1}}\int\limits_{1}^{B-1}\frac{dv}{\sqrt{v^2-1}}\\
       &= \int\limits_{1}^B\frac{\cosh^{-1}(B-1)du}{\sqrt{u^2-1}}\\
       &= \int\limits_{1}^B\frac{\log(B-1+\sqrt{(B-1)^2-1})du}{\sqrt{u^2-1}}\\
       &\leq \log(2B)\int\limits_{1}^B\frac{du}{\sqrt{u^2-1}}= \log(2B)\cosh^{-1}(B)\leq (\log(2B))^2.
   \end{align*}
   \item Case 2: $\frac{\lambda-2^j}{2d}\geq 1,$ and then $\frac{\log(\lambda-2d)}{\log 2}\geq j.$ Then, with $A= \frac{\lambda-2^j}{2d}$ and with $B=\frac{\lambda-2^{j-1}}{2d},$ we have that
   \end{item}
   \begin{align*}
         R_j &= \iint_{  \{(u,v): \frac{\lambda-2^j}{2d}\leq u+v\leq \frac{\lambda-2^{j-1}}{2d},\,u,v\geq 1 \} }\frac{dudv}{\sqrt{u^2-1}\sqrt{v^2-1}}\\
         &\leq \int\limits_{1}^A\frac{1}{\sqrt{u^2-1}}\int\limits_{1}^{B-u}\frac{dv du}{\sqrt{v^2-1}}+ \int\limits_{A}^B\frac{1}{\sqrt{u^2-1}}\int\limits_{1}^{B-u}\frac{dv du}{\sqrt{v^2-1}}\\
         &\leq \cosh^{-1}(A)\log(2B)+ (\cosh^{-1}(B)-\cosh^{-1}(A))\log(2B)\\
         &\leq (\log(2B))^2.
   \end{align*} The analysis above allows us to estimate
   \begin{align*}
      & \sum_{j\leq \frac{\log(\lambda-4d)}{\log 2}+1} 2^j R_j\\
      &=  \sum_{j\leq \frac{\log(\lambda-4d)}{\log 2}+1 : -\infty<j\leq 0  } 2^j R_j+\sum_{0 <j\leq \frac{\log(\lambda-4d)}{\log 2}+1 } 2^j R_j\\
       &\leq \sum_{j\leq \frac{\log(\lambda-4d)}{\log 2}+1 : -\infty<j\leq 0   } 2^j \log^2\left(\frac{\lambda-2^{j-1}}{d}\right)+\sum_{0 <j\leq \frac{\log(\lambda-4d)}{\log 2}+1 } 2^j R_j\\
         &\leq 2 \log^2\left(\frac{\lambda}{d}\right) +\sum_{1\leq j\leq \frac{\log(\lambda-4d)}{\log 2}+1 } 2^j R_j
         \\
         &\leq 2\log^2\left(\frac{\lambda}{d}\right)+\sum_{1\leq j\leq \frac{\log(\lambda-4d)}{\log 2}+1 } 2^j \left(\log\left(\frac{\lambda-2^{j-1}}{d}\right)\right)^2.
   \end{align*}
   Observe that in the interval $$1\leq x\leq \frac{\log(\lambda-4d)}{\log 2}+1 ,  $$ we have that 
$$ \frac{\lambda}{d}-\frac{\log(\lambda-4d)}{d\log 2}-\frac{1}{d}\leq \frac{\lambda-x}{d}\leq \frac{\lambda-1}{d},  $$ and we can assure that the lower bound $\frac{\lambda}{d}-\frac{\log(\lambda-4d)}{d\log 2}-\frac{1}{d}>1,$ if $\lambda $ is large enough. 
In this intervale the function $f(x)= \left(\log\left(\frac{\lambda-x}{d}\right)\right)^2$ is decreasing. 
   We can make the estimate of the last series as a Riemann summation on the interval $[1, \frac{\log(\lambda-4d)}{\log 2}+1]$ with the interior points being $x_j=2^{j-1}$ and with $\Delta x_j\leq  2^{j-1}-2^{j-2}=2^{j-2}$. Indeed,
   \begin{align*}
      & \sum_{1\leq j\leq \frac{\log(\lambda-4d)}{\log 2}+1 } 2^j \left(\log\left(\frac{\lambda-2^{j-1}}{d}\right)\right)^2 \\&=4\sum_{1\leq j\leq \frac{\log(\lambda-4d)}{\log 2}+1 } 2^{j-2} \left(\log\left(\frac{\lambda-2^{j-1}}{d}\right)\right)^2\\
       &\leq 4\int_1^{\frac{\log(\lambda-4d)}{\log 2}+1}\left(\log\left(\frac{\lambda-x}{d}\right)\right)^2dx.
\end{align*} A change of variables in the last integral allows us to compute
\begin{align*}       
        &\int_1^{\frac{\log(\lambda-4d)}{\log 2}+1}\left(\log\left(\frac{\lambda-x}{d}\right)\right)^2dx= d\int_{\frac{\lambda}{d}-\frac{\log(\lambda-4d)}{d\log 2}-\frac{1}{d}}^{\frac{\lambda}{d}-\frac{1}{d}}\left(\log\left(t\right)\right)^2dt\\
        &=d\left(\left(\frac{\lambda}{d}-\frac{1}{d}\right)\log^2\left(\frac{\lambda}{d}-\frac{1}{d}\right)-2\left(\frac{\lambda}{d}-\frac{1}{d}\right)\log\left(\frac{\lambda}{d}-\frac{1}{d}\right)+2\left(\frac{\lambda}{d}-\frac{1}{d}\right)\right)\\
        &- d\left(\left(\frac{\lambda}{d}-\frac{\log(\lambda-4d)}{d\log 2}-\frac{1}{d}\right)\log^2\left(\frac{\lambda}{d}-\frac{\log(\lambda-4d)}{d\log 2}-\frac{1}{d}\right)\right)\\
        &-d\left(-2\left(\frac{\lambda}{d}-\frac{\log(\lambda-4d)}{d\log 2}-\frac{1}{d}\right)\log\left(\frac{\lambda}{d}-\frac{\log(\lambda-4d)}{d\log 2}-\frac{1}{d}\right)+2\left(\frac{\lambda}{d}-\frac{\log(\lambda-4d)}{d\log 2}-\frac{1}{d}\right)\right)\\
        &=\left(\lambda-1\right)\log^2\left(\frac{\lambda}{d}-\frac{1}{d}\right)-2\left(\lambda-1\right)\log\left(\frac{\lambda}{d}-\frac{1}{d}\right)+2\left(\lambda-1\right)\\
        &-  \left(\lambda-\frac{\log(\lambda-4d)}{\log 2}-1\right)\log^2\left(\frac{\lambda}{d}-\frac{\log(\lambda-4d)}{d\log 2}-\frac{1}{d}\right)\\
        &+2\left(\lambda-\frac{\log(\lambda-4d)}{\log 2}-1\right)\log\left(\frac{\lambda}{d}-\frac{\log(\lambda-4d)}{d\log 2}-\frac{1}{d}\right)-2\left(\lambda-\frac{\log(\lambda-4d)}{\log 2}-1\right)\\
        &=\left(\lambda-1\right)\log^2\left(\frac{\lambda}{d}-\frac{1}{d}\right)-2\left(\lambda-1\right)\log\left(\frac{\lambda}{d}-\frac{1}{d}\right)\\
        &-  \left(\lambda-\frac{\log(\lambda-4d)}{\log 2}-1\right)\log^2\left(\frac{\lambda}{d}-\frac{\log(\lambda-4d)}{d\log 2}-\frac{1}{d}\right)\\
        &+2\left(\lambda-\frac{\log(\lambda-4d)}{\log 2}-1\right)\log\left(\frac{\lambda}{d}-\frac{\log(\lambda-4d)}{d\log 2}-\frac{1}{d}\right)+\frac{2\log(\lambda-4d)}{\log 2}.
   \end{align*} 
In consequence,
\begin{align*}
    I_1 &\leq  4\sum_{j\in \mathbb{Z}}\iint_{ \Omega_{j}} 2^jdkdy\\
    &\leq \frac{1}{\pi^2b^2}\sum_{j\leq \frac{\log(\lambda-4d)}{\log 2}+1} 2^j R_j\\
    &\leq \frac{2}{\pi^2b^2}  \log^2\left(\frac{\lambda}{d}\right) \\
    &+ \frac{4}{\pi^2b^2}\left(\left(\lambda-1\right)\log^2\left(\frac{\lambda}{d}-\frac{1}{d}\right)-2\left(\lambda-1\right)\log\left(\frac{\lambda}{d}-\frac{1}{d}\right)\right)\\
        &- \frac{4}{\pi^2b^2}\left(\lambda-\frac{\log(\lambda-4d)}{\log 2}-1\right)\log^2\left(\frac{\lambda}{d}-\frac{\log(\lambda-4d)}{d\log 2}-\frac{1}{d}\right)\\
        &+\frac{8}{\pi^2b^2}\left(\left(\lambda-\frac{\log(\lambda-4d)}{\log 2}-1\right)\log\left(\frac{\lambda}{d}-\frac{\log(\lambda-4d)}{d\log 2}-\frac{1}{d}\right)+\frac{\log(\lambda-4d)}{\log 2}\right).
\end{align*}
So, when $\lambda$ is large enough, one can estimate
\begin{align*}
     I_1 &\lesssim \frac{1}{\pi^2b^2}(\lambda \log^2(\lambda)-2\lambda\log(\lambda)-(\lambda- \log_2(\lambda))\log^2(\lambda-\log_2(\lambda))\\
     &+2(\lambda-\log_2(\lambda))\log(\lambda-\log_2(\lambda))+\log^2(\lambda)+\log_2(\lambda)).
\end{align*} 
\subsection{The dyadic analysis of the phase space: a lower bound}
For the analysis of $I_2,$ let us take $2d'=\max\{2/d_1,2/d_2\},$ and then 
\begin{align*}
I_2 &\geq 4 \iint_{\mathbb{R}^+_0\times \mathbb{R}^+_0}  \big(\lambda-\frac{2}{d_1}\cosh(2\pi bk)- \frac{2}{d_2}\cosh(2\pi by)\big)_+d kd y \\
&\geq 4 \iint_{\mathbb{R}^+_0\times \mathbb{R}^+_0} \big(\lambda-2d'\cosh(2\pi bk)-2 d'\cosh(2\pi by)\big)_+d kd y.
\end{align*} Let us decompose the set $\mathbb{R}^+_0\times \mathbb{R}^+_0$ in the family of dyadic zones  
\begin{equation}
    \widetilde\Omega_{j}=\{(k,y):2^{j-1}\leq \lambda-2d'\cosh(2\pi bk)- 2d'\cosh(2\pi by)\leq 2^j \},\quad j\in \mathbb{Z}.
\end{equation} 
We have the estimate
\begin{align*}
    I_2 &\geq 4  \iint_{\mathbb{R}^+_0\times \mathbb{R}^+_0} \big(\lambda-2d'\cosh(2\pi bk)- 2d'\cosh(2\pi by)\big)_+d kd y\asymp 4\sum_{j\in \mathbb{Z}}\iint_{ \widetilde\Omega_{j}} 2^jdkdy. 
\end{align*}
Note that for any $j\in \mathbb{Z},$ we have that 
\begin{align*}
    \widetilde \Omega_{j}=\{(y,k):\frac{\lambda-2^j}{2d'}\leq \cosh(2\pi bk)+ \cosh(2\pi by)\leq \frac{\lambda-2^{j-1}}{2d'} \}.
\end{align*} So, in our further analysis we fix $\lambda>4d'-1.$
Note that
$$\forall (y,k),\quad 2\leq \cosh(2\pi bk)+ \cosh(2\pi by). $$ In consequence, 
$$ \frac{\lambda-2^{j-1}}{2d'} <2\Longleftrightarrow j>\frac{\log(\lambda-4d')}{\log 2}+1,\textnormal{ and for these values of } j,\,\quad \widetilde\Omega_j=\{\emptyset\}.$$
Observe that the decomposition of $\mathbb{R}^+_0\times \mathbb{R}^+_0$ into the sets $\widetilde\Omega_j$ runs over $j\in \mathbb{Z}$ such that $j\leq \frac{\log(\lambda-4d')}{\log 2}+1. $

The change of variables $(u,v)=(\cosh(2\pi bk),\cosh(2\pi by))$ leads to 

\begin{align*} 
&4\sum_{j\in \mathbb{Z}}\iint_{ \widetilde\Omega_{j}} 2^jdkdy \\
&\asymp \frac{1}{\pi^2b^2} \sum_{j\leq \frac{\log(\lambda-4d')}{\log 2}+1} 2^j \iint_{  \{(u,v): \frac{\lambda-2^j}{2d'}\leq u+v\leq \frac{\lambda-2^{j-1}}{2d'},\,u,v\geq 1 \} }\frac{dudv}{\sqrt{u^2-1}\sqrt{v^2-1}}.
\end{align*} Let us analyse the integral
$$ R_j:= \iint_{  \{(u,v): \frac{\lambda-2^j}{2d'}\leq u+v\leq \frac{\lambda-2^{j-1}}{2d'},\,u,v\geq 1 \} }\frac{dudv}{\sqrt{u^2-1}\sqrt{v^2-1}}.$$
   \begin{item}
   \item Case 1: $\frac{\lambda-2^j}{2d'}<1.$ In this case $\frac{\log(\lambda-2d')}{\log 2}<j.$ As above let us define $A= \frac{\lambda-2^j}{2d'}$ and $B=\frac{\lambda-2^{j-1}}{2d'}.$
   Since $\lambda\rightarrow\infty,$ we can fix a constant $M_0\gg 2d',$ such that $\lambda\geq M_0+1.$ Then, for $j\leq 0,$  we can chose $\alpha<1,$ such that
   $$ B= \frac{2d'}{\lambda-2^{j-1}}\leq \frac{2d'}{\lambda-1}\leq \frac{2d'}{M_0}<\alpha<1.$$ Note that this choice of $\alpha$ allows us to compare $1/B<\alpha<1,$ or equivalently, $1<B\alpha<B,$ with $\alpha$ independent of $\lambda$ and of $j\neq 0.$ In general, observe that for any $j$ with $\widetilde\Omega_j\neq \{\emptyset\},$ that $B\geq 2,$ and then one can always choose $1/B\leq 1/2\leq \alpha<1,$ independent of $\lambda$ and of $j\leq \frac{\log(\lambda-4d)}{\log 2}+1,$ satisfying that $1\leq B\alpha<B.$
   
   Keeping in mind this choice of $\alpha$  
   we have that
   \begin{align*}
       R_j &= \iint_{  \{(u,v): \frac{\lambda-2^j}{2d'}\leq u+v\leq \frac{\lambda-2^{j-1}}{2d'},\,u,v\geq 1 \} }\frac{dudv}{\sqrt{u^2-1}\sqrt{v^2-1}}\\
       &\asymp \int\limits_{1}^B\frac{1}{\sqrt{u^2-1}}\int\limits_{1}^{B-u}\frac{dvdu}{\sqrt{v^2-1}}\\
       &\geq \int\limits_{1}^{\alpha B}\frac{\cosh^{-1}(B-u)du}{\sqrt{u^2-1}}\\
       &\geq \int\limits_{1}^{\alpha B}\frac{\cosh^{-1}(B-B\alpha)du}{\sqrt{u^2-1}}\\
       &= \cosh^{-1}(B(1-\alpha)) \cosh^{-1}(\alpha B)\geq \log^2(B\varkappa) ,
   \end{align*} where $0<\varkappa=\min(\alpha,1-\alpha).$ Note that the choice $\alpha=1/2,$ optimises the value of $\varkappa=1/2.$ In this case we have estimated
   \begin{equation*}
        R_j\gg  \log^2(B/2).
   \end{equation*}
   \item Case 2: $\frac{\lambda-2^j}{2d'}\geq 1,$ and then $\frac{\log(\lambda-2d')}{\log 2}\geq j.$ With $A= \frac{\lambda-2^j}{2d'}$ and with $B=\frac{\lambda-2^{j-1}}{2d'},$ we have that
   \end{item}
   \begin{align*}
         R_j &= \iint_{  \{(u,v): \frac{\lambda-2^j}{2d'}\leq u+v\leq \frac{\lambda-2^{j-1}}{2d'},\,u,v\geq 1 \} }\frac{dudv}{\sqrt{u^2-1}\sqrt{v^2-1}}\\
         &\asymp \int\limits_{1}^A\frac{1}{\sqrt{u^2-1}}\int\limits_{1}^{B-u}\frac{dv du}{\sqrt{v^2-1}}+ \int\limits_{A}^B\frac{1}{\sqrt{u^2-1}}\int\limits_{1}^{B-u}\frac{dv du}{\sqrt{v^2-1}}\\
          &=  \int\limits_{1}^B\frac{1}{\sqrt{u^2-1}}\int\limits_{1}^{B-u}\frac{dv du}{\sqrt{v^2-1}}\\
         & = \int\limits_{1}^B\frac{\cosh^{-1}(B-u)du}{\sqrt{u^2-1}}.
   \end{align*} Observe that for any $j$ with $\widetilde\Omega_j\neq \{\emptyset\},$ one has that $B\geq 2,$ and then one can always choose $1/B\leq 1/2<\alpha<1$ independent of $\lambda$ and of $j\leq \frac{\log(\lambda-4d')}{\log 2}+1,$ satisfying $1\leq B\alpha<B.$The analysis above allows us to estimate
$$  \int\limits_{1}^B\frac{\cosh^{-1}(B-u)du}{\sqrt{u^2-1}}\geq   \int\limits_{1}^{\alpha B}\frac{\cosh^{-1}(B-u)du}{\sqrt{u^2-1}}\geq \cosh^{-1}(B(1-\alpha)) \cosh^{-1}(\alpha B)$$
   $$\geq \log^2(B/2).$$ In consequence,
   \begin{align*}
      & \sum_{j\leq \frac{\log(\lambda-4d')}{\log 2}+1} 2^j R_j\\
      &=  \sum_{j\leq \frac{\log(\lambda-4d')}{\log 2}+1 : -\infty<j\leq 0  } 2^j R_j+\sum_{0 <j\leq \frac{\log(\lambda-4d')}{\log 2}+1 } 2^j R_j\\
       &= \sum_{j\leq \frac{\log(\lambda-4d')}{\log 2}+1 : -\infty<j\leq 0   } 2^j \log^2\left(\frac{\lambda-2^{j-1}}{4d'}\right)+\sum_{0 <j\leq \frac{\log(\lambda-4d')}{\log 2}+1 } 2^j R_j\\
         &\geq  \log^2\left(\frac{\lambda-1}{4d'}\right) +\sum_{1\leq j\leq \frac{\log(\lambda-4d')}{\log 2}+1 } 2^j R_j
         \\
         &\gg  \log^2\left(\frac{\lambda-1}{4d'}\right) +\sum_{1\leq j\leq \frac{\log(\lambda-4d)}{\log 2}+1 } 2^j \left(\log\left(\frac{\lambda-2^{j-1}}{4d'}\right)\right)^2.
   \end{align*}As above, we can make the estimate of the last series as a Riemann summation. Hence,
   \begin{align*}
       \sum_{1\leq j\leq \frac{\log(\lambda-4d)}{\log 2}+1 } 2^j \left(\log\left(\frac{\lambda-2^{j-1}}{4d'}\right)\right)^2 &=4\sum_{1\leq j\leq \frac{\log(\lambda-4d')}{\log 2}+1 } 2^{j-2} \left(\log\left(\frac{\lambda-2^{j-1}}{4d'}\right)\right)^2\\
       &\asymp \int_1^{\frac{\log(\lambda-4d')}{\log 2}+1}\left(\log\left(\frac{\lambda-x}{4d'}\right)\right)^2dx.
\end{align*} A change of variables in the last integral allows us to compute
\begin{align*}       
        &= 4d'\int_{\frac{\lambda}{4d'}-\frac{\log(\lambda-4d')}{4d'\log 2}-\frac{1}{4d'}}^{\frac{4\lambda}{d'}-\frac{1}{4d'}}\left(\log\left(t\right)\right)^2d't\\
        &=4\left(\lambda-1\right)\log^2\left(\frac{\lambda}{4d'}-\frac{1}{4d'}\right)-8\left(\lambda-1\right)\log\left(\frac{\lambda}{4d'}-\frac{1}{4d'}\right)\\
        &-  4\left(\lambda-\frac{\log(\lambda-4d')}{\log 2}-1\right)\log^2\left(\frac{\lambda}{d'}-\frac{\log(\lambda-4d')}{4d'\log 2}-\frac{1}{d'}\right)\\
        &+8\left(\lambda-\frac{\log(\lambda-4d')}{\log 2}-1\right)\log\left(\frac{\lambda}{d'}-\frac{\log(\lambda-4d')}{4d'\log 2}-\frac{1}{d'}\right)+\frac{8\log(\lambda-4d')}{\log 2}.
   \end{align*} 
In consequence,
\begin{align*}
    I_2 &\geq  4\sum_{j\in \mathbb{Z}}\iint_{ \Omega_{j}} 2^jdkdy\\
    &\asymp \frac{1}{\pi^2b^2}\sum_{j\leq \frac{\log(\lambda-4d')}{\log 2}+1} 2^j R_j\\
    &\gg\frac{1}{\pi^2b^2}  \log^2\left(\frac{\lambda-1}{4d'}\right) \\
    &+ \frac{1}{\pi^2b^2}\left(\left(\lambda-1\right)\log^2\left(\frac{\lambda}{4d'}-\frac{1}{4d'}\right)-2\left(\lambda-1\right)\log\left(\frac{\lambda}{4d'}-\frac{1}{4d'}\right)\right)\\
        &- \frac{1}{\pi^2b^2}\left(\lambda-\frac{\log(\lambda-4d')}{\log 2}-1\right)\log^2\left(\frac{\lambda}{4d'}-\frac{\log(\lambda-4d')}{4d'\log 2}-\frac{1}{4d'}\right)\\
        &+\frac{2}{\pi^2b^2}\left(\left(\lambda-\frac{\log(\lambda-4d')}{\log 2}-1\right)\log\left(\frac{\lambda}{4d'}-\frac{\log(\lambda-4d')}{4d'\log 2}-\frac{1}{4d'}\right)+\frac{\log(\lambda-4d')}{\log 2}\right).
\end{align*} In consequence, when $\lambda\rightarrow \infty,$ we acomplish the two conditions on $\lambda$ imposed above, namely that, $\lambda>4d-1,4d'-1,$ and  we allow for the estimate
\begin{align*}
    I_2 &\gg \frac{1}{\pi^2b^2}(\lambda \log^2(\lambda)-2\lambda\log(\lambda)-(\lambda- \log_2(\lambda))\log^2(\lambda-\log_2(\lambda))\\
     &+2(\lambda-\log_2(\lambda))\log(\lambda-\log_2(\lambda))+\log^2(\lambda)+\log_2(\lambda)).
\end{align*}  
The proof of Theorem \ref{theo1} is complete.  
\end{proof}
\noindent {\bf Acknowledgement.}   The author has been supported by the FWO Fellowship
Grant No 1204824N and by the FWO Grant K183725N of the Belgian Research Foundation FWO.

\bibliographystyle{amsplain}

\end{document}